\pdfoutput=1
\documentclass[11pt]{article}
\usepackage[utf8]{inputenc}
\usepackage[margin=1in]{geometry}

\usepackage{enumitem}

\usepackage[numbers,sort]{natbib}

\usepackage{tcolorbox}

\newtcolorbox{algorithm}[1]
{
	adjusted title = {#1},
	fonttitle = \bfseries,
  	beforeafter skip = 12pt,
}

\usepackage{amsfonts,amssymb,amsmath,amsthm}
\usepackage{mathtools}
\usepackage{thmtools}

\declaretheorem{theorem}
\declaretheorem[numberwithin=section]{lemma}
\declaretheorem[sibling=lemma]{corollary}

\declaretheorem[style=definition]{definition}

\usepackage[dvipsnames]{xcolor}
\usepackage[colorlinks=true,pdfpagemode=UseNone,urlcolor=RoyalBlue,linkcolor=RoyalBlue,citecolor=OliveGreen,pdfstartview=FitH]{hyperref}

\usepackage{cleveref}

\newcommand{\Unif}{\mathrm{Unif}}

\newcommand{\OPT}{\mathrm{OPT}}
\newcommand{\closure}{\mathrm{cl}}

\renewcommand{\Pr}{\mathbf{Pr}}
\newcommand{\E}{\mathbf{E}}

\newcommand{\e}{\mathrm{e}}

\newcommand{\ForwardKernel}{\mathsf{K}}
\newcommand{\ReverseKernel}{\mathsf{R}}

\title{Settling the Matroid Secretary Problem}
\author{Zhiyi Huang\thanks{The University of Hong Kong. Email: zhiyi@cs.hku.hk}}
\date{September 2026}

\begin{document}

\maketitle

\vspace{-20pt}
\begin{abstract}
    This paper settles the Matroid Secretary Problem with an $\e$-probability-competitive algorithm.
    The algorithm is ordinal and accesses arrived elements only through comparison and independence oracles, and has expected polynomial time and oracle complexity.
\end{abstract}

\section{Introduction}

The \emph{Matroid Secretary Problem} considers a finite matroid whose elements have nonnegative values unknown to the algorithm in advance.
The elements arrive independently and uniformly in $[0, 1]$.\footnote{It is folklore that this continuous-time model and the original discrete-time model, where elements arrive in uniformly random order with the number of elements known in advance, are equivalent up to a reduction.}
When an element arrives, the algorithm immediately and irrevocably decides whether to accept it.
The set of accepted elements must always be independent.
An online algorithm is $\Gamma$-competitive if the expected total value of its accepted elements is at least $\frac{1}{\Gamma}$ times the offline optimum.

The simplest case, known as the \emph{Secretary Problem}, considers accepting at most one element.
\citet{Dynkin-SovietMath-1963} resolved it in the 1960s: reject all elements before time $\frac{1}{\e}$, and then accept the first element that is the biggest among arrived ones.
This yields an optimal $\e$-competitive ratio.

Almost two decades ago, \citet{BabaioffIK-SODA-2007,BabaioffIKK-JACM-2018} introduced the Matroid Secretary Problem.
They gave constant-competitive algorithms for several families of matroids, and conjectured that a constant-competitive algorithm exists for all matroids.

After a substantial line of research, constant-competitive algorithms are known for many families of matroids, including uniform matroids~\cite{Kleinberg-SODA-2005,Dynkin-SovietMath-1963,HajiaghayiKP-EC-2004,BabaioffIKK-APPROX-2007}, laminar matroids~\cite{ImW-SODA-2011,HuangPZ-ESA-2024,BercziLSV-IPCO-2025}, transversal matroids~\cite{DimitrovP-ICALP-2008,KesselheimRTV-ESA-2013,SotoTV-MOR-2021}, graphic matroids~\cite{BabaioffIK-SODA-2007,BanihashemHKKMO-ESA-2025,DuttingPPP-2026}, and co-graphic matroids~\cite{Soto2013-SICOIMP-2013}.
For general matroids, \citet{BabaioffIK-SODA-2007} gave an $O(\log r)$-competitive algorithm, where $r$ is the rank, i.e., the size of maximal independent sets of the matroid.
This was later improved to $O(\sqrt{\log r})$ by \citet{ChakrabortyL-SODA-2012}, and to $O(\log \log r)$ by \citet{Lachish-FOCS-2014} and \citet{FeldmanSZ-SODA-2015}.

The literature has considered different model variants, e.g., whether the matroid is known in advance or revealed incrementally via an independence oracle over arrived elements (\emph{unknown matroid}), whether the algorithm observes the values or has only comparison-oracle access (\emph{ordinal}), and whether the goal is to maximize the expected total value or to accept each element in the offline optimum with probability at least $\frac{1}{\Gamma}$ (\emph{probability-competitive}).
In each pair, the latter option is more stringent than the former; e.g., a $\Gamma$-probability-competitive algorithm is $\Gamma$-competitive by linearity of expectation.
The results above hold in various combinations of these variants.

\paragraph{Recent Breakthroughs.}
Very recently, \citet{Singla-2026} made a major breakthrough, giving a $4$-probability-competitive ordinal algorithm for all matroids in the unknown matroid model, settling the conjecture of \citet{BabaioffIK-SODA-2007}.
\citet{BercziDLSV-2026} and \citet{AbdiBHM-2026} independently obtained ordinal algorithms with the optimal $\e$-probability-competitive ratio for linear matroids, a broad family encompassing regular matroids (including graphic and co-graphic matroids) and gammoids (including laminar and transversal matroids).
They refer to the existence of an $\e$-competitive algorithm for general matroids as the Strong Matroid Secretary Conjecture.
\citet{AbdiBHM-2026} also obtained a $64$-competitive algorithm for general matroids and a tight $2$-competitive single-sample matroid prophet inequality.
Before these results, $\e$-competitive algorithms were known only for a few matroid families, including uniform matroids~\cite{BabaioffIKK-APPROX-2007} and transversal matroids~\cite{KesselheimRTV-ESA-2013}.

\paragraph{Our Results.}
We resolve the Strong Matroid Secretary Conjecture in its most stringent form.

\begin{theorem}
    \label{thm:main}
    There is an $\e$-probability-competitive ordinal algorithm for the Matroid Secretary Problem in the unknown matroid model, with expected polynomial time and oracle complexity.
\end{theorem}

The design and analysis of our algorithm are guided by an auxiliary survival process, in which elements also arrive independently and uniformly in $[0, 1]$.
Consider the set $S_p$ of elements that arrive by time $p = \frac{1}{\e}$, and let $\OPT(S_p)$ be its optimal independent set.
For each $e \in \OPT(S_p)$, let $\sigma_e$ be the first time when it is no longer in the optimum of arrived elements;
let $\sigma_e = 1$ if $e$ is in the final optimum.
Finally, independently let each $e \in \OPT(S_p)$ survive with probability $\ln\frac{\sigma_e}{p}$.

Each element $e$ in the final optimum arrives by time $p$ with probability $p$, and in that case survives with probability $\ln\frac{\sigma_e}{p} = \ln\frac{1}{p} = 1$.
Hence, $e$ survives with probability exactly $\frac{1}{\e}$.
However, the process is not online: a survivor must arrive by time $p$, but its survival depends on later arrivals.

We design an online algorithm whose set of accepted elements is distributionally identical to the set of survivors in the above process.
More strongly, at any time $t \ge p$ and conditional on the set of elements arrived by then, the set of accepted elements has the same distribution as the survivors of the same process with horizon $t$ in place of $1$.
Like Dynkin's algorithm, our algorithm rejects all elements before time $p$.
Afterward, it accepts each arriving element that joins the current optimum with a carefully chosen probability so that the evolution dynamics match those of the survival process.

We do not know how to efficiently compute this acceptance probability, which is a ratio of two probabilities under survival processes.
Instead, we couple the algorithm with the survival process by maintaining a uniformly distributed latent vector.
The coupling implements the acceptance rule exactly, with expected polynomial time and oracle complexity.

\paragraph{Relation to Prior Work.}
The idea of controlling the entire distribution of the algorithm's outcome is inspired by both \citet{Singla-2026} and the stationary online contention resolution schemes of \citet*{AminianNN-EC-2026}.
In particular, \citet{Singla-2026} uses reversible online updates to preserve the distribution of a random configuration in the two-sided Game of Googol.
Our efficient implementation similarly maintains a uniformly distributed latent vector through forward and reverse kernels.
The algorithms of \citet{BercziDLSV-2026} and \citet{AbdiBHM-2026} for linear matroids use linear programs to prescribe marginal acceptance probabilities, bounding for every subspace the expected dimension of its intersection with the span of the accepted elements.
Our algorithm instead matches the entire distribution of the accepted set to that of the survivors of an explicit auxiliary process.

\section{Preliminaries}

\paragraph{Model.}
Consider a finite matroid $M = (E, \mathcal{I})$ and a total order $\succ$ over $E$.
We assume without loss of generality that there are no loops, i.e., every element is in some independent set, as loops can be safely omitted.
We say $e$ is bigger than $f$ if $e \succ f$.
Let $\OPT$ denote the optimal independent set with respect to $\succ$.
More generally, let $\OPT(S)$ be the optimal independent set restricted to a set of elements $S$, so that $\OPT = \OPT(E)$.

Each element $e \in E$ independently arrives at time $T_e \sim \Unif[0, 1]$.
An online algorithm has access to an independence oracle and a comparison oracle over the arrived elements.
The former answers whether a subset of arrived elements is independent, while the latter provides pairwise comparisons of arrived elements according to $\succ$.
When an element arrives, the algorithm immediately and irrevocably decides whether to accept it.
The set of accepted elements must be independent at all times.
The algorithm is $\Gamma$-probability-competitive if for every element $e \in \OPT$, it accepts $e$ with probability at least $\frac{1}{\Gamma}$.

\paragraph{Matroid Basics.}
We provide minimal background on matroids that will be used in this paper.
The \emph{rank} of a set $S \subseteq E$ is the size of the largest independent subset of $S$.
An element $e$ is \emph{spanned} by a set of elements $S \subseteq E$ if $S$ and $S+e$ have the same rank, where we write $S + e$ and $S - e$ for $S \cup \{e\}$ and $S \setminus \{e\}$ throughout the paper.
The optimal independent set $\OPT(S)$ can be computed greedily by iterating through the elements of $S$ in decreasing order and adding each element if it is not spanned by the elements before it.
The \emph{closure} $\closure(S)$ of a subset of elements $S \subseteq E$ is the set of elements spanned by $S$.
Closure is monotone, i.e., $S \subseteq T$ implies $\closure(S) \subseteq \closure(T)$.

\medskip

We further need the following standard properties of matroids.

\begin{lemma}[Inheritance of Optimality~\cite{Oxley-2011}]
    \label{lem:opt-inheritance}
    If $e \in S \subseteq T$ and $e \in \OPT(T)$, then $e \in \OPT(S)$.
\end{lemma}

\begin{lemma}[Closure Exchange Property~\cite{Oxley-2011}]
    \label{lem:exchange-property}
    For any $S$ and any $e, f$, we have
\[
        f \in \closure(S+e) \setminus \closure(S)
        \quad\Longleftrightarrow\quad
        e \in \closure(S+f) \setminus \closure(S).
    \]
\end{lemma}

\begin{lemma}[Irrelevance of Unspanned Elements]
    \label{lem:unspanned}
    If $e \notin \closure(T)$ and $X + f \subseteq T$, then $f \in \closure(X+e)$ if and only if $f \in \closure(X)$.
\end{lemma}

\begin{proof}
    The ``if'' direction follows from monotonicity.
    For the ``only if'' direction, suppose for contradiction that $f \in \closure(X+e) \setminus \closure(X)$.
By the exchange property (\Cref{lem:exchange-property}) and monotonicity, $e \in \closure(X+f) \subseteq \closure(T)$, a contradiction.
\end{proof}

\section{Auxiliary Survival Processes}
\label{sec:survival}

In the rest of the paper, let $p \in [\frac{1}{\e}, 1]$ be a fixed threshold.
The optimal choice will be $p = \frac{1}{\e}$, but the analysis works for any $p$ in this range.

\begin{definition}[Filtration and Spanning Time]
    Fix any $t \in [p, 1]$.
    A \emph{filtration} $F = (F_u)_{u \in [p, t]}$ is a nondecreasing right-continuous family of subsets of $E$.
    For any element $e$, its \emph{spanning time} with respect to filtration $F$ is
\[
        \sigma_{e,F} = \inf \bigl\{ u \in [p, t] : e \in \closure(F_u) \bigr\},
    \]
    with the convention that the infimum of the empty set equals $t$.
\end{definition}

We write $\varnothing$ and $e$ for the constant filtrations $F_u = \varnothing$ and $F_u = \{e\}$.
We omit the subscript $F$ from $\sigma_{e,F}$ when it is clear from context.

For example, let $p = \frac{1}{\e}$, $t = 1$, and define each element $e$'s spanning time with respect to a filtration where $F_u$ is the set of elements bigger than $e$ arrived by time $u$.
Dynkin's algorithm for the $1$-uniform matroid selects the first element arriving between time $p$ and its spanning time.

\begin{definition}[Survival Process]
    \label{def:survival-process}
    For any $t \in [p, 1]$ and any subset of elements $S$, let $\mathcal{D}_{S,t}$ be the distribution of the set of survivors $A \subseteq S$ in the following process.
    \begin{enumerate}
        \item Independently draw an arrival time $T_e \sim \Unif[0, t]$ for each element $e \in S$.
        \item For any $u \in [p, t]$, let the set of elements arrived by time $u$ be $S_u = \bigl\{ e \in S : T_e \le u \bigr\}$.
        \item For any $e \in S$, let the subset of bigger elements arrived by time $u$ be $S_{e,u} = \bigl\{ f \in S_u : f \succ e \bigr\}$.
        \item  For any $e \in S$, let $\sigma_e$ be its spanning time with respect to the filtration $(S_{e,u})_{u \in [p, t]}$.
\item Independently let each $e \in S_p$ survive with probability $\theta_e = \ln \frac{\sigma_e}{p}$.
    \end{enumerate}
\end{definition}

Letting $p \in [\frac{1}{\e}, 1]$ ensures that the survival probability $\theta_e = \ln \frac{\sigma_e}{p}$ is at most $1$.
Almost surely, every surviving element $e$ arrives by time $p$, i.e., $e \in S_p$, and belongs to $\OPT(S_p)$, as otherwise, by the greedy characterization of $\OPT$, $\sigma_e = p$ and $\theta_e = 0$.
Moreover, the survival probability is particularly simple for elements in $\OPT(S)$.

\begin{lemma}
    \label{lem:survival-competitive}
    For any $t \in [p, 1]$, any subset of elements $S$, and any $e \in \OPT(S)$, we have
    \[
        \Pr_{A \sim \mathcal{D}_{S,t}} \bigl[ e \in A \bigr] = \frac{p}{t} \ln \frac{t}{p}.
    \]
\end{lemma}

\begin{proof}
    Since $e \in \OPT(S)$, it is not spanned by the bigger elements of $S$.
    By monotonicity, $\sigma_e = t$.
    Hence, the probability that $e$ survives equals $\frac{p}{t}$, the probability that $e$ arrives by time $p$ and thus is in $\OPT(S_p)$ by inheritance of optimality (\Cref{lem:opt-inheritance}), multiplied by $\theta_e = \ln \frac{t}{p}$, the conditional probability that $e$ survives.
\end{proof}

To analyze $\mathcal{D}_{S,t}$, we will consider the biggest element $f \in S$ and condition on its arrival time.
If $f$ arrives at time $u$, it is in the spanning filtration of all other elements throughout $[\max\{p, u\},t]$.
This motivates a survival process with a fixed element added to the filtration after a given time.
Inductively, we need to consider survival processes with a background filtration.
We next introduce this richer family of survival processes and establish their relationship.

\begin{definition}[Survival Process with Background Filtration]
    For any $t \in [p, 1]$, any subset of elements $S$, and any background filtration $F = (F_u)_{u \in [p, t]}$ disjoint from $S$, let $\mathcal{D}_{S,t,F}$ be the distribution of the set of survivors $A \subseteq S$ in the process consisting of the first three steps in \Cref{def:survival-process} and the next two steps.
    \begin{enumerate}[start=4]
        \item For each $e \in S$, let $\sigma_e$ be its spanning time with respect to $(F_u \cup S_{e, u})_{u \in [p, t]}$.
        \item Independently let each $e \in S_p$ survive with probability $\theta_e = \ln \frac{\sigma_e}{p}$.
    \end{enumerate}
\end{definition}

The original survival process is the special case with an empty filtration.
That is, we have
\[
    \mathcal{D}_{S,t} = \mathcal{D}_{S,t,\varnothing}, \qquad \forall t \in [p, 1].
\]

Given any filtration $F = (F_u)_{u \in [p, t]}$, we write $F + e = (F_u \cup \{e\})_{u \in [p, t]}$ for the filtration with element $e$ added.
The next lemma relates the survival processes with filtrations $F$ and $F+e$.

\begin{lemma}
    \label{lem:distribution-change}
    For any $t \in [p, 1]$, any subset of elements $S$, any filtration $F$ disjoint from $S$, any $e \notin S$, and any subset $A \subseteq S$, we have
\[
        \mathcal{D}_{S,t,F}(A) \ge \frac{p}{\sigma_{e,F}} \mathcal{D}_{S,t,F+e}(A).
    \]
\end{lemma}

This is the main lemma that enables online simulation.
We first show two simpler facts about how the spanning time changes as the filtration gets bigger.

\begin{lemma}
    \label{lem:spanning-time-monotone}
    Suppose $F = (F_u)_{u \in [p, t]}$ and $F^+ = (F^+_u)_{u \in [p, t]}$ satisfy $F_u \subseteq F^+_u$ for $u \in [p, t]$.
    Then, for any element $e$, we have $\sigma_{e,F} \ge \sigma_{e,F^+}$.
\end{lemma}

\begin{proof}
    By the monotonicity of closure, a filtration of bigger sets can only span $e$ earlier.
\end{proof}

\begin{lemma}
    \label{lem:spanning-time-change}
    Let $e, f$ be distinct elements, and $F = (F_u)_{u \in [p, t]}$ be a filtration.
    Suppose $\sigma_{f,F} > \sigma_{f,F+e}$.
    Then, we have
    \[
        \sigma_{e,F} = \sigma_{f,F},
        \qquad
        \sigma_{e,F+f} = \sigma_{f,F+e}.
    \]
    Moreover, for a filtration $F(u)$ with $f$ added to $F$ from time $u \in [p, t]$ onward, we have
    \[
        \sigma_{e,F(u)} = \min \bigl\{ \sigma_{e,F}, \max \bigl\{ \sigma_{f,F+e}, u\bigr\} \bigr\}.
    \]
\end{lemma}

\begin{proof}
    At any time $u \in [p, t)$, $e \in \closure(F_u+f) \setminus \closure(F_u)$ if and only if $f \in \closure(F_u+e) \setminus \closure(F_u)$ by the exchange property (\Cref{lem:exchange-property}).
    The time interval for the former is $[\sigma_{e,F+f}, \sigma_{e,F})$ and that for the latter is $[\sigma_{f,F+e}, \sigma_{f,F})$.
    Hence, we have $\sigma_{e,F} = \sigma_{f,F}$ and $\sigma_{e,F+f} = \sigma_{f,F+e}$.
    This covers the expression of $\sigma_{e,F(u)}$ for $u = p$ and $u = t$.

    Next, consider adding $f$ to $F$ at some time $u \in (p, t)$.
    For $v < u$, we have $F_v(u) = F_v$, which spans $e$ if and only if $v \ge \sigma_{e,F}$.
    For $v \ge u$, we have $F_v(u) = F_v + f$, which spans $e$ if and only if $v \ge \sigma_{e,F+f}$.
    Putting together gives the claimed expression for $\sigma_{e,F(u)}$.
\end{proof}

\begin{proof}[Proof of \Cref{lem:distribution-change}]
    Throughout this proof, we omit the fixed $t$ from the subscripts of $\mathcal{D}$, e.g., writing $\mathcal{D}_{S,G}$ for $\mathcal{D}_{S,t,G}$.
    We prove by induction on the size of $S$.
    The base case of $S = \varnothing$ is trivial.

    Next, consider a nonempty $S$ and suppose the lemma holds for all smaller subsets.
    Let $f$ be the biggest element in $S$.
    Hence, its spanning time is $\sigma_{f,F}$ under $F$ and $\sigma_{f,F+e}$ under $F+e$, regardless of the other arrivals.
    Let $F(u)$ be the filtration obtained from adding element $f$ to $F$ from time $u$ onward.
    Note that $F(p) = F + f$.
    By \Cref{lem:spanning-time-monotone}, we have $\sigma_{f,F} \ge \sigma_{f,F+e}$.

    \paragraph{Case 1: $\sigma_{f,F} = \sigma_{f,F+e}$.}
Denote this common value as $\sigma_f$.
    By \Cref{lem:spanning-time-change} with $e$ and $f$ interchanged, $\sigma_{e,F} > \sigma_{e,F+f}$ would imply $\sigma_{f,F} > \sigma_{f,F+e}$.
    Hence, by monotonicity (\Cref{lem:spanning-time-monotone}), we have $\sigma_{e,F} = \sigma_{e,F+f}$.
Since $F \subseteq F(u) \subseteq F+f$, monotonicity also gives $\sigma_{e,F(u)} = \sigma_{e,F}$ for every $u \in [p, t]$.
    Consider the probabilities of realizing $A \subseteq S-f$.
    We have
\[
        \mathcal{D}_{S,F}(A)
        = \frac{p}{t} \Bigl(1 - \ln\frac{\sigma_f}{p}\Bigr) \mathcal{D}_{S-f,F(p)}(A) + \frac{1}{t} \int_p^t \mathcal{D}_{S-f,F(u)}(A) du.
    \]

    Here, $\frac{p}{t}$ is the probability that $f$ arrives by time $p$, in which case $f$ must not survive to realize $A$, captured by the factor $1 - \ln\frac{\sigma_f}{p}$.
    The density of $f$'s arrival time $u$ is $\frac{1}{t}$.
    Moreover, since $f$ is the biggest element in $S$, it is in the filtration of other elements once it arrives.
    The set of surviving elements other than $f$ follows $\mathcal{D}_{S-f,F(\max\{p,u\})}$ conditional on $f$'s arrival time $u$.

    The same argument gives
\[
        \mathcal{D}_{S,F+e}(A)
        = \frac{p}{t}  \Bigl(1 - \ln\frac{\sigma_f}{p}\Bigr) \mathcal{D}_{S-f,F(p)+e}(A) + \frac{1}{t} \int_p^t \mathcal{D}_{S-f,F(u)+e}(A) du.
    \]

    By the induction hypothesis, we have
\begin{equation}
        \label{eqn:case-1}
        \mathcal{D}_{S-f,F(u)}(A)
        \ge \frac{p}{\sigma_{e,F(u)}} \mathcal{D}_{S-f,F(u)+e}(A) = \frac{p}{\sigma_{e,F}} \mathcal{D}_{S-f,F(u)+e}(A).
    \end{equation}

    Combining the above proves the claimed inequality.

    A similar argument goes through for the probabilities of realizing $A \cup \{f\}$ for any $A \subseteq S-f$.

    \paragraph{Case 2: $\sigma_{f,F} > \sigma_{f,F+e}$.}
    Let $\sigma_0 = \sigma_{f,F} = \sigma_{e,F}$ and $\sigma_1 = \sigma_{f,F+e} = \sigma_{e,F+f}$ by \Cref{lem:spanning-time-change}.
    First, consider the probabilities of realizing $A \cup \{f\}$ where $A \subseteq S-f$.
    They are
\begin{align*}
        \mathcal{D}_{S,F}(A \cup \{f\})
         &
        = \frac{p}{t} \ln \frac{\sigma_0}{p} \mathcal{D}_{S-f,F+f}(A),
        \qquad
        \mathcal{D}_{S,F+e}(A \cup \{f\})
        = \frac{p}{t} \ln \frac{\sigma_1}{p} \mathcal{D}_{S-f,F+f+e}(A).
    \end{align*}

    Omitting the $\frac{1}{t}$ density term, by $\sigma_1 \le \sigma_0$ (\Cref{lem:spanning-time-monotone}), the coefficient for $\mathcal{D}_{S-f,F+f}$ satisfies
\begin{equation}
        \label{eqn:case-2-a-1}
        p \ln \frac{\sigma_0}{p} \ge p \,\frac{\sigma_1}{\sigma_0} \ln \frac{\sigma_1}{p}.
    \end{equation}

    By the induction hypothesis, we have
\begin{equation}
        \label{eqn:case-2-a-2}
        \mathcal{D}_{S-f,F+f}(A)
        \ge \frac{p}{\sigma_1} \mathcal{D}_{S-f,F+f+e}(A).
    \end{equation}

    Combining these two inequalities proves the claimed inequality.
    Readers may find it more natural to first apply the induction hypothesis and then compare coefficients.
    However, the order above corresponds better to a later algorithmic version of the lemma (\Cref{lem:constructive-domination}).
    The same remark applies to the rest of the proof.

    \bigskip

    Next, consider the probabilities of realizing $A \subseteq S-f$.
    They are
\begin{align*}
        \mathcal{D}_{S,F}(A)
         &
        = \frac{p}{t} \Bigl(1 - \ln\frac{\sigma_0}{p}\Bigr) \mathcal{D}_{S-f,F(p)}(A) + \frac{1}{t} \int_p^t \mathcal{D}_{S-f,F(u)}(A) du, \\
        \mathcal{D}_{S,F+e}(A)
         &
        = \frac{p}{t} \Bigl(1 - \ln\frac{\sigma_1}{p}\Bigr) \mathcal{D}_{S-f,F(p)+e}(A) + \frac{1}{t} \int_p^t \mathcal{D}_{S-f,F(u)+e}(A) du.
    \end{align*}

    For the first term, which corresponds to $f$'s arriving by time $p$, we have
\begin{equation}
        \label{eqn:case-2-b-1}
        \mathcal{D}_{S-f,F(p)}(A)
        \ge \frac{p}{\sigma_{e,F(p)}} \mathcal{D}_{S-f,F(p)+e}(A)
        = \frac{p}{\sigma_1} \mathcal{D}_{S-f,F(p)+e}(A).
    \end{equation}

    For each integrand $\mathcal{D}_{S-f,F(u)}(A)$ in the second term, which corresponds to $f$ arriving at time $u \in (p,t]$, \Cref{lem:spanning-time-change} implies $\sigma_{e,F(u)} \le \sigma_0$.
    We split the coefficient $1$ into two nonnegative parts $\frac{\sigma_{e,F(u)}}{\sigma_0}$ and $1-\frac{\sigma_{e,F(u)}}{\sigma_0}$.
    For the first part, the induction hypothesis gives
\begin{equation}
        \label{eqn:case-2-c-1}
        \mathcal{D}_{S-f,F(u)}(A)
        \ge \frac{p}{\sigma_{e,F(u)}} \mathcal{D}_{S-f,F(u)+e}(A).
    \end{equation}
This multiplied by $\frac{\sigma_{e,F(u)}}{\sigma_0}$ matches the integrand in $\mathcal{D}_{S,F+e}(A)$ up to the required $\frac{p}{\sigma_0}$ factor.

    For the second part, we apply the induction hypothesis twice, first with filtration $F(u)$ and element $e$, and then with filtration $F(u)+e$ and element $f$.
    Since $f$ joins the filtration $F(u)+e$ at time $u$, we have $\sigma_{f,F(u)+e}=\min\{u,\sigma_1\}$.
    Together with $F(u)+e+f=F(p)+e$, we obtain
\begin{equation}
        \label{eqn:case-2-c-2}
        \mathcal{D}_{S-f,F(u)}(A)
        \ge \frac{p}{\sigma_{e,F(u)}} \mathcal{D}_{S-f,F(u)+e}(A)
        \ge \frac{p}{\sigma_{e,F(u)}}
        \frac{p}{\min\{u,\sigma_1\}} \mathcal{D}_{S-f,F(p)+e}(A).
    \end{equation}

    Weight \eqref{eqn:case-2-b-1} by $\frac{p}{t}\bigl(1-\ln\frac{\sigma_0}{p}\bigr)$, and integrate \eqref{eqn:case-2-c-2} over $u\in[p,t]$ with weight $\frac{1}{t}\bigl(1-\frac{\sigma_{e,F(u)}}{\sigma_0}\bigr)$.
    The resulting coefficient of $\mathcal{D}_{S-f,F(p)+e}(A)$ is
\[
        \frac{p^2}{t\sigma_1}\Bigl(1-\ln\frac{\sigma_0}{p}\Bigr)
        + \frac1t\int_p^{\sigma_1}\left(\frac{p}{\sigma_1}-\frac{p}{\sigma_0}\right)\frac{p}{u}\,du
        + \frac1t\int_{\sigma_1}^{\sigma_0}\left(\frac{p}{u}-\frac{p}{\sigma_0}\right)\frac{p}{\sigma_1}\,du = \frac{p}{\sigma_0} \frac{p}{t}\Bigl(1-\ln\frac{\sigma_1}{p}\Bigr),
    \]
    where we plug in $\sigma_{e,F(u)} = \min\{\sigma_0, \max\{\sigma_1, u\}\}$ from \Cref{lem:spanning-time-change}, with the integrand vanishing for $u \ge \sigma_0$.
    This matches the first term in $\mathcal{D}_{S,F+e}(A)$ up to the required $\frac{p}{\sigma_0}$ factor.
\end{proof}

\begin{corollary}
    \label{cor:main}
    For any subsets $A \subseteq S$, any $e \notin S$, and any $t \in [p, 1]$, we have
\[
        \mathcal{D}_{S,t}(A) \ge \frac{p}{t} \cdot \mathcal{D}_{S,t,e}(A).
    \]
\end{corollary}

\begin{proof}
    This follows from \Cref{lem:distribution-change} and $\sigma_{e,\varnothing} = t$, since the empty set does not span $e$.
\end{proof}

\section{Online Algorithm}

Let $p \in [\frac{1}{\e}, 1]$ be the same threshold as in the previous section.

\begin{algorithm}{Survival Simulation}
Maintain the set of arrived elements $S$ and the set of accepted elements $A$.
\begin{enumerate}
        \item Reject arrivals before time $p$.
        \item For an element $e$ arriving at time $t > p$:
\begin{enumerate}
                  \item If $e \notin \OPT(S+e)$, reject.
                  \item Otherwise, accept with probability $\frac{p}{t} \cdot \frac{\mathcal{D}_{S,t,e}(A)}{\mathcal{D}_{S,t}(A)}$.
\end{enumerate}
\end{enumerate}
\end{algorithm}

The algorithm is well-defined because the acceptance probability is at most $1$ by \Cref{cor:main}.
We follow the convention that $\frac{0}{0} = 0$, so the algorithm rejects $e$ when $\mathcal{D}_{S,t,e}(A) = \mathcal{D}_{S,t}(A) = 0$.
Moreover, $A + e$ is independent when $\mathcal{D}_{S,t,e}(A) > 0$, and therefore the set of accepted elements is always independent.

\begin{theorem}
    \label{thm:simulation}
    For any $t \in [p, 1]$ and conditional on the set $S$ of elements arrived by time $t$, the set of accepted elements $A$ follows distribution $\mathcal{D}_{S,t}$.
\end{theorem}

\begin{proof}
    Let $\mathcal{P}_{S,t}$ denote the distribution of the set of accepted elements conditional on the set $S$ of elements arrived by time $t$.
    We prove by induction on the size of $S$ that $\mathcal{P}_{S,t} = \mathcal{D}_{S,t}$ for $t \in [p, 1]$.
    The base case of $S = \varnothing$ is trivial.
    Next, consider any nonempty set $S$ of size $n$, and suppose the claim holds for smaller subsets.
    At time $t = p$, both $\mathcal{P}_{S,t}$ and $\mathcal{D}_{S,t}$ give an empty set almost surely.
    It remains to verify that their evolution dynamics over time are identical.

    \paragraph{Survival Process.}
To analyze $\frac{d}{dt} \mathcal{D}_{S,t}(A)$, we consider $\mathcal{D}_{S,t+h}$ for a small $h>0$, in comparison with $\mathcal{D}_{S,t}$.
    Below we use $\simeq$ to omit $O(h^2)$ terms and discuss the first-order effects.

    \bigskip
    \noindent
    \emph{New Arrival:~}
    Each $e \in S$ is the unique element arriving between $t$ and $t+h$ with probability $\simeq \frac{h}{t}$.
    We omit the $O(h^2)$-probability event of two or more arrivals in this interval.
    If $e \in A$, $A$ cannot be realized.
    If $e \notin A$, $A$ is realized with probability $\mathcal{D}_{S-e,t}(A) + O(h)$, where $O(h)$ accounts for the potential change of survival probability for elements in $S - e$.
    Thus, each element $e \in S \setminus A$ in this case contributes $\simeq \frac{h}{t} \mathcal{D}_{S-e,t}(A)$ to the probability of realizing $A$, and $\frac{1}{t} \mathcal{D}_{S-e,t}(A)$ to the derivative.

    \bigskip
    \noindent
    \emph{Change of Density:~}
The density of arrival vectors in $[0,t]^S$ reduces by a $\simeq 1 - \frac{nh}{t}$ factor, decreasing the probability of realizing $A$ by $\simeq \frac{nh}{t} \mathcal{D}_{S,t}(A)$, and contributing $- \frac{n}{t} \mathcal{D}_{S,t}(A)$ to the derivative.

    \bigskip
    \noindent
    \emph{Increase of Survival Probability:~}
    Conditional on all arrivals in $[0, t]$, each element $e \in \OPT(S)$ arrives by time $p$ with probability $\frac{p}{t}$.
    In that case, $e$ is not spanned by the bigger elements of $S$, so its spanning time equals the horizon.
    Extending the horizon to $t+h$ increases its survival probability by $\ln\frac{t+h}{p} - \ln\frac{t}{p} \simeq \frac{h}{t}$.\footnote{Elements not in $\OPT(S)$ are spanned by time $t$, and thus do not benefit from the extended horizon.}
Moreover, by \Cref{lem:unspanned}, adding $e$ to the filtration of any bigger element does not change its spanning time, while $e$ is already in the filtration of every smaller element from time $p$.
    Thus, the surviving elements other than $e$ follow $\mathcal{D}_{S-e,t,e}$ at horizon $t$, while extending the horizon to $t+h$ changes this conditional distribution by only $O(h)$, yielding an $O(h^2)$ term when multiplied by the $O(h)$ increase in $e$'s survival probability.
If $e \in A$, it increases the probability of realizing $A$ by $\simeq \frac{h}{t} \frac{p}{t} \mathcal{D}_{S-e,t,e}(A-e)$, contributing $\frac{p}{t^2} \mathcal{D}_{S-e,t,e}(A-e)$ to the derivative.
    If $e \notin A$, it decreases the probability of realizing $A$ by $\simeq \frac{h}{t} \frac{p}{t} \mathcal{D}_{S-e,t,e}(A)$, contributing $- \frac{p}{t^2} \mathcal{D}_{S-e,t,e}(A)$ to the derivative.

\bigskip

    Write $O = \OPT(S)$ for brevity in the rest of the proof.
    The above first-order effects give
\[
        \frac{d}{dt} \mathcal{D}_{S,t}(A) = - \frac{n}{t} \mathcal{D}_{S,t}(A) + \frac{1}{t} \biggl( \sum_{e \in O \cap A} \frac{p}{t} \mathcal{D}_{S-e,t,e}(A-e) - \sum_{e \in O \setminus A} \frac{p}{t} \mathcal{D}_{S-e,t,e}(A) + \sum_{e \in S \setminus A} \mathcal{D}_{S-e,t}(A) \biggr).
    \]

    \paragraph{Online Algorithm.}
    To analyze $\frac{d}{dt} \mathcal{P}_{S,t}(A)$, we consider $\mathcal{P}_{S,t+h}$ for a small $h>0$, in comparison with $\mathcal{P}_{S,t}$.
    Unlike in the survival process, the set of accepted elements changes only at arrivals.
    The algorithm's new-arrival effects account for both the new-arrival and the survival-probability effects in the survival process.

    \bigskip
    \noindent
    \emph{New Arrival:~}
    Each $e \in S$ is the unique element arriving between $t$ and $t+h$ with probability $\simeq \frac{h}{t}$.
    We omit the $O(h^2)$-probability event of two or more arrivals in this interval.
    When $e$ arrives in this interval, the set of accepted elements before $t$ follows $\mathcal{P}_{S-e,t} = \mathcal{D}_{S-e,t}$ by the induction hypothesis.

    If $e \in \OPT(S) \setminus A$, the algorithm rejects it with probability $1 - \frac{p}{t} \frac{\mathcal{D}_{S-e,t,e}(A)}{\mathcal{D}_{S-e,t}(A)} + O(h)$, where $O(h)$ is an error term for using $t$ instead of $e$'s arrival time in the expression of rejection probability.
    This case contributes
\[
        \simeq \frac{h}{t} \mathcal{D}_{S-e,t}(A) \Big( 1 - \frac{p}{t} \frac{\mathcal{D}_{S-e,t,e}(A)}{\mathcal{D}_{S-e,t}(A)} \Big) = \frac{h}{t} \mathcal{D}_{S-e,t}(A) - \frac{h}{t} \frac{p}{t} \mathcal{D}_{S-e,t,e}(A)
    \]
    to the probability of realizing $A$, and $\frac{1}{t} \mathcal{D}_{S-e,t}(A) - \frac{p}{t^2} \mathcal{D}_{S-e,t,e}(A)$ to the derivative.

    If $e \in \OPT(S) \cap A$, the algorithm accepts it with probability $\frac{p}{t} \frac{\mathcal{D}_{S-e,t,e}(A-e)}{\mathcal{D}_{S-e,t}(A-e)} + O(h)$, contributing $\frac{h}{t} \frac{p}{t} \mathcal{D}_{S-e,t,e}(A-e)$ to the probability of realizing $A$, and $\frac{p}{t^2} \mathcal{D}_{S-e,t,e}(A-e)$ to the derivative.

    If $e \in S \setminus \OPT(S)$, the algorithm rejects $e$ by definition.
    If further $e \notin A$, this case contributes $\frac{h}{t} \mathcal{D}_{S-e,t}(A)$ to the probability of realizing $A$, and $\frac{1}{t} \mathcal{D}_{S-e,t}(A)$ to the derivative.

    \bigskip
    \noindent
    \emph{Change of Density:~}
    The density of arrival vectors in $[0,t]^S$ reduces by a $\simeq 1 - \frac{nh}{t}$ factor.
    This decreases the probability of realizing $A$ by $\simeq \frac{nh}{t} \mathcal{P}_{S,t}(A)$, contributing $- \frac{n}{t} \mathcal{P}_{S,t}(A)$ to the derivative.

    \bigskip

    The first-order effects combine to give
\[
        \frac{d}{d t} \mathcal{P}_{S,t}(A)
        = - \frac{n}{t} \mathcal{P}_{S,t}(A) + \frac{1}{t} \biggl( \sum_{e \in O \cap A} \frac{p}{t} \mathcal{D}_{S-e,t,e}(A-e) - \sum_{e \in O \setminus A} \frac{p}{t} \mathcal{D}_{S-e,t,e}(A) + \sum_{e \in S \setminus A} \mathcal{D}_{S-e,t}(A) \biggr).
    \]

    Subtracting the two evolution equations shows that $t^n(\mathcal{P}_{S,t}(A)-\mathcal{D}_{S,t}(A))$ is constant.
    The common initial condition at $t=p$ makes this constant zero, i.e., $\mathcal{P}_{S,t}(A) = \mathcal{D}_{S,t}(A)$.
\end{proof}

\begin{proof}[Proof of \Cref{thm:main} (Except Running Time)]
    By \Cref{thm:simulation} at time $t = 1$, the final set of accepted elements follows distribution $\mathcal{D}_{E,1}$.
    By \Cref{lem:survival-competitive}, each element in $\OPT$ is accepted with probability $p \ln \frac{1}{p}$.
    Letting $p = \frac{1}{\e}$ gives the claimed probability-competitive ratio of $\e$.
\end{proof}

\section{Efficient Implementation}
\label{sec:polytime}

The acceptance rule in the previous section depends on the ratio of two probabilities under the survival processes.
We do not know how to evaluate the ratio efficiently.
This section gives an explicit coupling of the survival processes and the online algorithm's acceptance process.
By doing so, we avoid evaluating the ratio and obtain an efficient implementation.
Throughout this section, keep the same $p\in[\frac{1}{\e},1]$ as before.

\subsection{Latent Space for Survival Processes}
\label{subsec:latent-space}

Fix any subset $S$, time $t$, and filtration $F$ disjoint from $S$.
Write $X_S = (X_e)_{e\in S}\in[0,t]^S$ for the latent vector for $\mathcal{D}_{S,t,F}$;
we omit the subscript when $S = E$.
Let
\[
    \mu_{S,t}=\Unif([0,t]^S).
\]

Interpret each $X_e$ as follows.
If $X_e>p$, then $e$ arrives at time $X_e$.
If $X_e\le p$, then $e$ arrives by time $p$.
Since the precise arrival time does not matter, set the clipped arrival time to be $p$ and recycle $\frac{X_e}{p} \sim \Unif[0, 1]$ for the survival randomness.
Formally, define the clipped arrival time
\[
    T'_e=\max\{p,X_e\}.
\]

For each $e\in S$, let $\sigma_e(X_{S-e})$ be its spanning time with respect to the filtration
\[
    \left(
    F_u\cup\{f\in S:f\succ e,\ T'_f\le u\}
    \right)_{u\in[p,t]}.
\]

Define the set of survivors as
\[
    \mathcal A_{S,t,F}(X_S) =
    \left\{
    e\in S:
    X_e\le p\ln\frac{\sigma_e(X_{S-e})}{p}
    \right\}.
\]

By the above definition, we have the following lemma.

\begin{lemma}[Latent representation]
    \label{lem:poly-latent}
    If $X_S \sim \mu_{S,t}$, then $\mathcal A_{S,t,F}(X_S)\sim \mathcal{D}_{S,t,F}$.
\end{lemma}

\subsection{Efficient Online Algorithm}
\label{subsec:efficient-online}

We strengthen \Cref{lem:distribution-change} from a pointwise inequality on survivor-set probabilities to an explicit transport on the latent spaces.  The proof is deferred to \Cref{subsec:constructive-proof}.

\begin{lemma}[Algorithmic version of \Cref{lem:distribution-change}]
    \label{lem:constructive-domination}
    For any subset $S$, time $t$, filtration $F$ disjoint from $S$, and element $e\notin S$, there is a forward kernel
\[
        \ForwardKernel_{S,t,F,e}:[0,t]^S\to[0,t]^S\cup\{\perp\}
    \]
    from the latent space of $\mathcal{D}_{S,t,F}$ to that of $\mathcal{D}_{S,t,F+e}$, where $\perp$ denotes failure, and a reverse kernel
    \[
        \ReverseKernel_{S,t,F,e}:[0,t]^S\to[0,t]^S
    \]
    from the latent space of $\mathcal{D}_{S,t,F+e}$ back to that of $\mathcal{D}_{S,t,F}$, with the following properties.
    \begin{enumerate}[label=(\roman*),leftmargin=*]
        \item \textbf{Adjoint identity.}
              The following equality holds as measures on $[0,t]^S\times[0,t]^S$
              \begin{equation}
                  \label{eq:poly-adjoint}
                  \mu_{S,t}(dX_S)\,
                  \ForwardKernel_{S,t,F,e}(X_S,dY_S)
                  =
                  \frac{p}{\sigma_{e,F}}\,\mu_{S,t}(dY_S)\,
                  \ReverseKernel_{S,t,F,e}(Y_S,dX_S).
              \end{equation}
              In particular, for a uniform input, $\ForwardKernel$ succeeds with probability exactly $\frac{p}{\sigma_{e,F}}$, and conditional on success its output is uniform.

        \item \textbf{Survivor preservation.}  Every pair $(X_S,Y_S)$ produced by a successful forward call or by a reverse call satisfies
              \[
                  \mathcal A_{S,t,F}(X_S) = \mathcal A_{S,t,F+e}(Y_S).
              \]

        \item \textbf{Monotone reverse kernel.}  Every such pair and their common survivor set $A$ satisfy
              \[
                  X_f\le Y_f \quad(f\in A),
                  \qquad
                  X_f\ge Y_f \quad(f\notin A).
              \]

        \item \textbf{Complexity.}  For a uniform input, $\ForwardKernel_{S,t,\varnothing,e}$ and $\ReverseKernel_{S,t,\varnothing,e}$ have expected polynomial time and oracle complexity in the size of $S$, in the unknown matroid model with only $S+e$ revealed.
\end{enumerate}
\end{lemma}

For later use, define the pointwise success probability
\[
    r_{S,t,F,e}(X_S)
    =
    \Pr\!\left[
        \ForwardKernel_{S,t,F,e}(X_S)\ne\perp
        \right].
\]

Integrating \eqref{eq:poly-adjoint} over $Y_S \in [0, t]^S$ and $X_S \in [0, t]^S$ restricted to $\mathcal A_{S,t,F}(X_S) = A$ and using survivor preservation, we have
\[
    \E\!\left[
        r_{S,t,F,e}(X_S)
        \mathrel{\big|}
        \mathcal A_{S,t,F}(X_S)=A
        \right]
    \cdot \mathcal{D}_{S,t,F}(A)
    =
    \frac{p}{\sigma_{e,F}}
    \mathcal{D}_{S,t,F+e}(A),
\]
which implies the original \Cref{lem:distribution-change}.

We now construct the online map from physical arrival randomness to the latent state.  Define
\[
    g(t)=p\ln\frac{t}{p},
    \qquad
    g'(t)=\frac{p}{t},
    \qquad
    g^{-1}(u)=p \, \e^{\frac{u}{p}},
\]
where $g$ maps a latent variable in $[p, 1]$ to one in $[0, p \ln \frac{1}{p}]$, and $g^{-1}$ maps backward.

\begin{algorithm}{Efficient Survival Simulation}

    Maintain the sets of arrived and accepted elements $S$, $A$, and latent vector $X_S \in [0,1]^S$.

    \begin{enumerate}[leftmargin=*]
        \item \textbf{Before time $p$:~}  Reject every $e$ and let its arrival time be latent variable $X_e \leftarrow T_e$.

        \item \textbf{Arrival of $e$ at $T_e > p$:}
\begin{enumerate}
                  \item If $e\notin\OPT(S+e)$, reject $e$ and let its arrival time be latent variable $X_e \leftarrow T_e$.

                  \item If $e\in\OPT(S+e)$, let $Y_S \leftarrow \ForwardKernel_{S,T_e,\varnothing,e}(X_S)$.
                        \begin{enumerate}
                            \item If it fails, leave $X_S$ unchanged, reject $e$, and let $X_e \leftarrow T_e$.
                            \item Otherwise, accept $e$, and let $X_S\leftarrow Y_S$, $X_e\leftarrow g(T_e)$.
                        \end{enumerate}
              \end{enumerate}

        \item \textbf{Reverse events:}
              \begin{enumerate}
                  \item For every $f\in\OPT(S)\setminus A$ with latent variable $X_f\le p$, let its reverse time be
                        \[
                            R_f=g^{-1}(X_f).
                        \]
                        Recompute the reverse time after every arrival or reverse event.
                  \item At the reverse time $R_f$ of element $f$, let $X_{S-f} \leftarrow \ReverseKernel_{S-f,R_f,\varnothing,f}(X_{S-f})$, $X_f \leftarrow R_f$.
\end{enumerate}
    \end{enumerate}
\end{algorithm}

We first explain the intuition behind the changes of latent variables.
Step 2-(b)-ii corresponds to the transportation of probability mass from $A$ to $A+e$ due to $e$'s arrival in $T_e \in [t,t+h]$ in the algorithm, and due to $e$'s survival probability increasing by $\ln \frac{t+h}{p} - \ln \frac{t}{p}$ in the survival process.
The latter corresponds to $\frac{X_e}{p} \in [\ln \frac{t}{p}, \ln \frac{t+h}{p}]$ and thus $X_e \in [g(t), g(t+h)]$ with our latent representation.
Hence, $X_e \leftarrow g(T_e)$ correctly maps the former infinitesimal interval to the latter.

Conversely, when the time horizon extends to the reverse time $R_f = g^{-1}(X_f)$ of $f$, $f$ would become a survivor with latent variable $X_f$.
Since $f$ is already rejected, we remap the latent variables using the reverse kernel.
As we shall see, this mapping complements the failure case of the forward kernel in step 2-(b)-i to restore a uniform distribution over the latent variables.

Next, we will prove four lemmas, whose combination certifies the efficiency and correctness of the algorithm with the optimal $\e$-probability-competitiveness.

\begin{lemma}
    \label{lem:latent-to-set}
    Efficient Survival Simulation maintains $A=\mathcal A_{S,t,\varnothing}(X_S)$ for $t \in [p, 1]$ almost surely.
\end{lemma}

\begin{proof}
    At time $p$, both $A$ and $\mathcal{A}_{S,t,\varnothing}(X_S)$ are empty.
    It remains to consider the updates.

    First, consider an element $e$ arriving at time $T_e >p$.
    If $e$ is rejected and placed at $X_e \leftarrow T_e$, it is late and cannot survive.
    Its release at the horizon $T_e$ does not change any earlier spanning time, so the set of survivors stays unchanged.

    Next, suppose $e\in\OPT(S+e)$ and the forward kernel succeeds with output $Y_S$.
    Since $e \in \OPT(S+e)$, it is not spanned by the elements of $S$ bigger than $e$.
    By \Cref{lem:unspanned}, adding $e$ to the filtrations of the bigger elements does not change their spanning times, and $X_e = g(T_e) \le p$ places $e$ in the filtrations of the smaller elements from time $p$.
    Hence, $e$ may be considered as background filtration for all existing elements.
    By survivor preservation (\Cref{lem:constructive-domination}, part (ii)), the old survivor set is unchanged when $X_S$ is replaced by $Y_S$.
    Since $e\in\OPT(S+e)$, its spanning time is $T_e$, and $X_e=g(T_e)=p\ln\frac{T_e}{p}$, so $e$ lies exactly on its survival boundary and is added to the survivor set.

    Finally, only elements in $\OPT(S)$ have moving survival thresholds between events.
    Existing survivors remain survivors.
    Right before $f \in \OPT(S) \setminus A$ could become a new survivor, a reverse event would be triggered.
    Before that moment, $X_f \le p$ places $f$ in the filtrations of the smaller elements from time $p$, and since $f \in \OPT(S)$, adding $f$ to the filtrations of the bigger elements does not change their spanning times by \Cref{lem:unspanned}.
    Hence, $f$ acts as background filtration for all existing elements.
Applying the reverse kernel $\ReverseKernel_{S-f,R_f,\varnothing,f}$ removes this background while preserving their survivor set, and setting $X_f = R_f > p$ keeps $f$ rejected.
\end{proof}

\begin{lemma}
    \label{lem:reverse-events-bound}
    During the time interval while the set of arrived elements is $S$, at most $|S|$ reverse events occur.
\end{lemma}

\begin{proof}
    By the monotonicity of the reverse kernel (\Cref{lem:constructive-domination}, part (iii)), a reverse event can only increase the coordinates of rejected elements, and the set of accepted elements remains unchanged by survivor preservation (\Cref{lem:constructive-domination}, part (ii)).
    Further, the triggering coordinate itself is moved above $p$.
    Hence, no rejected element can trigger twice.
\end{proof}

\begin{lemma}\label{lem:poly-uniform}
    At every time $t\in[p,1]$ and conditional on any set of arrived elements $S$, we have
    \[
        X_S \sim \mu_{S,t},
        \qquad
        A \sim \mathcal{D}_{S,t}.
    \]
\end{lemma}

\begin{proof}
    It suffices to prove the claim for $X_S$, and then the claim for $A$ follows from \Cref{lem:poly-latent,lem:latent-to-set}.
    We consider an equivalent process on the full latent vector $X=(X_e)_{e\in E}$, initialized with $X_e = T_e \sim \Unif[0, 1]$ for every $e\in E$ at time $0$.
    The latent vector goes through two types of events over time: (i) the arrival of $e \in \OPT(S)$ at $T_e = t$ and the forward kernel succeeds, and (ii) the reverse event of $e \in \OPT(S)$ at time $R_e = t$.

    We prove the stronger claim that $X\sim\Unif([0,1]^E)$ at every time by verifying the probability-mass balance of the process.
    Then, the condition on the set of arrived items $S$ at time $t$ simply corresponds to $X_S \in [0, t]^S$ and $X_{E-S} \in (t,1]^{E-S}$, which preserves uniformity of $X_S$.
    For this verification, we evaluate these two types of events under the candidate uniform law at time $t$.
    We show that, for each $e$, the probability mass transported by the events cancels.

    In event (i), the successful forward kernel calls transport away
\[
        r_{S-e,t,\varnothing,e}(X_{S-e})\,dX_{S-e} dX_e.
    \]
    Since the forward-kernel call happens as long as $X_e = T_e = t$, $X_{S-e} \le t$, and $X_{E-S} > t$, the conditional distribution of $X_{S-e}$ is uniform.
    By the adjoint identity (\Cref{lem:constructive-domination}, part (i)), the successful calls transport this mass to $\frac{p}{t}\,dY_{S-e}dX_e$.
    Further, the algorithm sets $Y_e \leftarrow g(X_e)$ so $dY_e = g'(X_e) dX_e = \frac{p}{t} dX_e$.
    Thus, the transported mass is $dY_{S-e}dY_e$.

    In event (ii), let $Y$ be the full latent vector before the reverse event.
    The reverse event of $e$ occurs when the increasing boundary $g(t)$ reaches $Y_e$ from below while $e$ is still rejected.
    Under the candidate uniform law, the incoming mass with $Y_{S-e}\le t$ and $Y_{E-S}>t$ is uniform in $Y_{S-e}\in[0,t]^{S-e}$.
    Again, by the adjoint identity (\Cref{lem:constructive-domination}, part (i)) and $\sigma_{e,\varnothing} = t$, integrating over the uniform $Y_{S-e}$, the reverse event transports back
\[
        \int dY_{S-e} \ReverseKernel_{S-e,t,\varnothing,e}(Y_{S-e},dX_{S-e}) dY_e = \frac{t}{p} r_{S-e,t,\varnothing,e}(X_{S-e})\,dX_{S-e} dY_e.
    \]
This is exactly $r_{S-e,t,\varnothing,e}(X_{S-e})\,dX_{S-e} dX_e$, since $X_e = t = g^{-1}(Y_e)$, and thus, $dY_e = \frac{p}{t} dX_e$.

    In sum, event (i) moves mass $r_{S-e,t,\varnothing,e}(X_{S-e})\,dX_{S-e}\,dX_e$ from the late slice $X_e \in [t, t+dt]$ to the early slice $X_e \in [g(t), g(t+dt)]$, where it has density $1$.
    Event (ii) moves all mass of the early slice, which lies above the survival threshold $g(t)$, back to the late slice, where it has density $r_{S-e,t,\varnothing,e}(X_{S-e})$.
    Hence, the two transports cancel under the candidate uniform law, which satisfies the mass-balance equations.
    There are finitely many arrivals and reverse events (\Cref{lem:reverse-events-bound}).
    These equations therefore determine the law uniquely by successively transporting the initial mass through the events.
    Since the initial law is uniform, $X$ remains uniform at every time.
\end{proof}

\begin{lemma}
    \label{lem:poly-complexity}
Efficient Survival Simulation has expected polynomial time and oracle complexity.
\end{lemma}

\begin{proof}
    The computation of $\OPT(S)$ has polynomial time and oracle complexity for matroids by the greedy algorithm.
    It remains to analyze the kernel calls.

    There are $|E|$ arrivals, and thus, at most $|E|$ forward kernel calls.
    Further, \Cref{lem:reverse-events-bound} implies that there are at most $O(|E|^2)$ reverse kernel calls.
    By the same arguments as in the proof of \Cref{lem:poly-uniform}, every kernel input is uniform conditional on time $t$, the arrived set of elements $S$, and the triggering element $e$.
Part (iv) of \Cref{lem:constructive-domination} bounds the expected kernel cost averaged over the incoming event mass by a polynomial in $|S| \le |E|$.
\end{proof}

\begin{proof}[Proof of \Cref{thm:main}]
    By \Cref{lem:poly-uniform} at time $t = 1$, the final set accepted by Efficient Survival Simulation follows distribution $\mathcal{D}_{E,1}$.
    By \Cref{lem:survival-competitive}, setting $p = \frac{1}{\e}$ gives acceptance probability $\frac{1}{\e}$ for every element in $\OPT$.
    The algorithm uses only comparisons and independence queries on arrived elements, and has expected polynomial time and oracle complexity by \Cref{lem:poly-complexity}.
\end{proof}

\subsection[Proof of Lemma~\ref{lem:constructive-domination}]{Proof of \Cref{lem:constructive-domination}}
\label{subsec:constructive-proof}

We first define the kernels recursively, following the cases in the proof of \Cref{lem:distribution-change}.
For a nonempty $S$, let $f$ be its biggest element.
Let $F(u)$ be the filtration obtained by adding $f$ to $F$ from time $u$ onward, so $F(p)=F+f$.
When $\sigma_{f,F}>\sigma_{f,F+e}$, use the same notation
\[
    \sigma_0=\sigma_{f,F} = \sigma_{e,F},
    \qquad
    \sigma_1=\sigma_{f,F+e} = \sigma_{e,F+f}
\]
as in that proof, and write the survival thresholds of $f$ under $F$ and $F+e$ as
\[
    a=p\ln\frac{\sigma_0}{p},
    \qquad
    b=p\ln\frac{\sigma_1}{p}.
\]

\paragraph{Forward Kernel.}
The cases below follow the proof of \Cref{lem:distribution-change}.
Relaxations of coefficients correspond to discarding probability mass by returning $\perp$, while applications of the induction hypothesis correspond to recursive calls to the forward kernel.
We explicitly associate each step with a corresponding inequality in the proof of \Cref{lem:distribution-change}.
If any recursive forward call fails, return $\perp$ immediately.

\begin{algorithm}{Forward Kernel $\ForwardKernel_{S,t,F,e}(X_S)$}
    \textbf{Base Cases:~} If $\sigma_{e,F}=p$, return $X_S$.
    If $S=\varnothing$, return the empty vector with probability $\frac{p}{\sigma_{e,F}}$, and $\perp$ with probability $1-\frac{p}{\sigma_{e,F}}$.\\[2ex]
Let $f$ be the biggest element in $S$, and $F(u)$ the filtration with $f$ added from time $u$ onward.

    \smallskip

    \textbf{Case 1:~} $\sigma_{f,F}=\sigma_{f,F+e}$.
    Set $u=\max\{p,X_f\}$, and let
    \[
        Y_{S-f}\leftarrow\ForwardKernel_{S-f,t,F(u),e}(X_{S-f}),
        \qquad Y_f\leftarrow X_f.
        \tag{Eqn.~\eqref{eqn:case-1}}
    \]

    \textbf{Case 2:~} $\sigma_{f,F}>\sigma_{f,F+e}$.
\begin{enumerate}[label=(\alph*)]
        \item \textbf{$f$ survives: $X_f\le a$.}
              \begin{enumerate}[label=(\roman*)]
                  \item If $\frac{\sigma_1}{\sigma_0}b < X_f \le a$, return $\perp$. \hspace*{\fill} (Eqn.~\eqref{eqn:case-2-a-1})
                  \item Otherwise, let
                        \[
                            Y_{S-f}\leftarrow\ForwardKernel_{S-f,t,F(p),e}(X_{S-f}),
                            \quad
                            Y_f\leftarrow\frac{\sigma_0}{\sigma_1}X_f. \tag{Eqn.~\eqref{eqn:case-2-a-2}}
                        \]
              \end{enumerate}
        \item \textbf{$f$ is rejected and early: $a<X_f\le p$.}
              Let
              \[
                  Y_{S-f}\leftarrow\ForwardKernel_{S-f,t,F(p),e}(X_{S-f}),
                  \qquad Y_f\leftarrow b+\frac{p-b}{p-a}(X_f-a).
                  \tag{Eqn.~\eqref{eqn:case-2-b-1}}
              \]

        \item \textbf{$f$ is late: $X_f=u>p$.}
\begin{enumerate}[label=(\roman*)]
                  \item With probability $\frac{\sigma_{e,F(u)}}{\sigma_0}$, let
                        \[
                            Y_{S-f}\leftarrow\ForwardKernel_{S-f,t,F(u),e}(X_{S-f}),
                            \qquad
                            Y_f\leftarrow X_f.
                            \tag{Eqn.~\eqref{eqn:case-2-c-1}}
                        \]
                  \item Otherwise, let
                        \[
                            Y_{S-f}\leftarrow\ForwardKernel_{S-f,t,F(u)+e,f}(\ForwardKernel_{S-f,t,F(u),e}(X_{S-f})),
                            \qquad
                            Y_f\sim\Unif(b,p).
                            \tag{Eqn.~\eqref{eqn:case-2-c-2}}
                        \]
              \end{enumerate}
    \end{enumerate}
\end{algorithm}

The maps $Y_f\leftarrow\frac{\sigma_0}{\sigma_1}X_f$ in Case 2(a)(ii) and $Y_f\leftarrow b+\frac{p-b}{p-a}(X_f-a)$ in Case 2(b) are the unique nondecreasing maps transporting $\Unif[0,\frac{\sigma_1}{\sigma_0}b]$ to $\Unif[0,b]$ and $\Unif(a,p]$ to $\Unif(b,p]$, respectively.
They satisfy $X_f\le Y_f$ for surviving $f$ and $X_f\ge Y_f$ for rejected $f$.
In Case 2(c)(ii), an independent draw $Y_f\sim\Unif(b,p)$ automatically satisfies $Y_f<X_f$, since $X_f>p$.

\paragraph{Reverse Kernel.}
To reverse the successful transports of the forward kernel, we invert the scalar maps and reverse the order of recursive calls.
We explicitly associate each step in the reverse kernel with its corresponding forward case.
We prefix case numbers with $\ForwardKernel$ or $\ReverseKernel$ to indicate the forward or reverse kernel, respectively, as in $\ForwardKernel$-2(c) and $\ReverseKernel$-2(b).

An early rejected output, handled by $\ReverseKernel$-2(b), can come from $\ForwardKernel$-2(b) or $\ForwardKernel$-2(c)(ii).
Their contributions determine the following mixture.
For $\sigma_{f,F}>\sigma_{f,F+e}$, define the probability
\[
    \lambda=\frac{\sigma_0}{\sigma_1}\frac{p-a}{p-b}.
\]

We have $0\le\lambda\le1$, because $s\mapsto s(1-\ln\frac{s}{p})$ is nonnegative and decreasing on $[p,1]$.

Further, define the unnormalized density from $\ForwardKernel$-2(c)(ii) as
\[
    \rho(u) =
    \underbrace{\vphantom{\Bigg|} \left(1-\frac{\sigma_{e,F(u)}}{\sigma_0}\right)}_{\substack{\text{\vphantom{Ag}branch}\\\text{\vphantom{Ag}probability}}}
    \cdot
    \underbrace{\vphantom{\Bigg|} \frac{p}{\sigma_{e,F(u)}} \cdot \frac{p}{\sigma_{f,F(u)+e}}}_{\substack{\text{\vphantom{Ag}success probabilities}\\\text{\vphantom{Ag}of two recursive calls}}},
    \qquad p<u<\sigma_0.
\]

\begin{algorithm}{Reverse Kernel $\ReverseKernel_{S,t,F,e}(Y_S)$}
    \textbf{Base Cases:~} If $\sigma_{e,F}=p$, return $Y_S$.
    If $S=\varnothing$, return the empty vector.\\[2ex]
Let $f$ be the biggest element in $S$, and $F(u)$ the filtration with $f$ added from time $u$ onward.

    \smallskip

    \textbf{Case 1:~} $\sigma_{f,F}=\sigma_{f,F+e}$.
    Set $u=\max\{p,Y_f\}$, and let
    \[
        X_{S-f}\leftarrow\ReverseKernel_{S-f,t,F(u),e}(Y_{S-f}),
        \qquad X_f\leftarrow Y_f.
        \tag{$\ForwardKernel$-1}
    \]

    \textbf{Case 2:~} $\sigma_{f,F}>\sigma_{f,F+e}$.
    \begin{enumerate}[label=(\alph*)]
        \item \textbf{$f$ survives: $Y_f\le b$.}
              Let
              \[
                  X_{S-f}\leftarrow\ReverseKernel_{S-f,t,F(p),e}(Y_{S-f}),
                  \qquad X_f\leftarrow\frac{\sigma_1}{\sigma_0}Y_f.
                  \tag{$\ForwardKernel$-2(a)(ii)}
              \]

        \item \textbf{$f$ is rejected and early: $b<Y_f\le p$.}
              \begin{enumerate}[label=(\roman*)]
                  \item With probability $\lambda$, let
                        \[
                            X_{S-f}\leftarrow\ReverseKernel_{S-f,t,F(p),e}(Y_{S-f}),
                            \qquad X_f\leftarrow a+\frac{p-a}{p-b}(Y_f-b).
                            \tag{$\ForwardKernel$-2(b)}
                        \]
                  \item Otherwise, draw $u\in(p,\sigma_0)$ with density proportional to $\rho(u)$, and let
                        \[
                            X_{S-f}\leftarrow\ReverseKernel_{S-f,t,F(u),e}(\ReverseKernel_{S-f,t,F(u)+e,f}(Y_{S-f})),
                            \qquad X_f\leftarrow u.
                            \tag{$\ForwardKernel$-2(c)(ii)}
                        \]
              \end{enumerate}

        \item \textbf{$f$ is late: $Y_f=u>p$.}
              Let
              \[
                  X_{S-f}\leftarrow\ReverseKernel_{S-f,t,F(u),e}(Y_{S-f}),
                  \qquad X_f\leftarrow Y_f.
                  \tag{$\ForwardKernel$-2(c)(i)}
              \]
    \end{enumerate}
\end{algorithm}

\begin{proof}[Proof of \Cref{lem:constructive-domination}]
    We prove the lemma by induction on the size of $S$.
The base cases $S=\varnothing$ and $\sigma_{e,F}=p$ are immediate.

    Next consider any nonempty set $S$ and suppose the lemma holds for smaller sets.
    We first recall the following spanning-time identities.
    Let $f$ be the biggest element in $S$.
    If $\sigma_{f,F}=\sigma_{f,F+e}$, then \Cref{lem:spanning-time-change}, with $e$ and $f$ interchanged, implies $\sigma_{e,F+f}=\sigma_{e,F}$.
    Since $F\subseteq F(u)\subseteq F+f$, by \Cref{lem:spanning-time-monotone} we have
    \begin{equation}
        \label{eq:poly-equal-case-sigma}
        \sigma_{e,F(u)}=\sigma_{e,F},
        \qquad u\in[p,t].
    \end{equation}
If $\sigma_{f,F} > \sigma_{f,F+e}$, \Cref{lem:spanning-time-change} gives
    \begin{equation}
        \label{eq:poly-exchange-identities}
        \sigma_{e,F}=\sigma_0,
        \qquad
        \sigma_{e,F(u)}=\min\{\sigma_0,\max\{\sigma_1,u\}\},
        \qquad
        \sigma_{f,F(u)+e}=\min\{u,\sigma_1\}.
    \end{equation}

    \paragraph{Adjoint Identity, Survivor Preservation, and Monotonicity.}
We verify the first three properties together by matching each successful forward branch with its corresponding reverse branch.
    In each case, we apply the inductive adjoint identity to the recursive calls and check the remaining scalar factors.
    For survivor preservation and monotonicity, the recursive calls use the filtrations determined by the clipped arrival time of $f$ on the corresponding side, so induction handles the coordinates in $S-f$.

    \bigskip
    \underline{$\ForwardKernel$-1 $\leftrightarrow$ $\ReverseKernel$-1}.~
The coordinate of the biggest element is copied, and \eqref{eq:poly-equal-case-sigma} gives the same factor $\frac{p}{\sigma_{e,F}}$ in the recursive adjoint identity.
    By induction, the adjoint identity holds.
    The coordinate $X_f=Y_f$ has the same survival threshold on both sides, so survivor preservation and monotonicity also hold.

    \bigskip
    \underline{$\ForwardKernel$-2(a)(ii) $\leftrightarrow$ $\ReverseKernel$-2(a)}.~
The forward map $Y_f \leftarrow \frac{\sigma_0}{\sigma_1}X_f$ multiplies the density by $\frac{\sigma_1}{\sigma_0}$, while the recursive call succeeds with probability $\frac{p}{\sigma_1}$ on a uniform input.
    Their product is $\frac{p}{\sigma_0}$ as required.
    $\ReverseKernel$-2(a) applies the inverse scalar map and the recursive reverse kernel.
    By induction, the adjoint identity holds.
    Further, we have $X_f\le\frac{\sigma_1}{\sigma_0}b\le a$ and $Y_f\le b$, so $f$ survives on both sides.
    Finally, the scalar map gives $X_f = \frac{\sigma_1}{\sigma_0}Y_f\le Y_f$ as required.

    \bigskip
    \underline{$\ForwardKernel$-2(b) $\leftrightarrow$ $\ReverseKernel$-2(b)(i)}.~
If $a=p$, neither branch is taken: the forward input interval is empty and the reverse branch has probability $\lambda=0$.
    Otherwise, by induction, the recursive forward-kernel call succeeds with probability $\frac{p}{\sigma_1}$ on a uniform input, and produces a uniform output conditional on success.
    The scalar map $Y_f \leftarrow b+\frac{p-b}{p-a}(X_f-a)$ multiplies the density by $\frac{p-a}{p-b}$.
    Thus, this branch contributes density relative to the uniform law $\frac{p}{\sigma_1}\frac{p-a}{p-b} =\frac{p}{\sigma_0}\lambda$,
    which equals the required adjoint factor $\frac{p}{\sigma_0}$ times the probability $\lambda$ of choosing $\ReverseKernel$-2(b)(i).
    So the adjoint identity follows.
    Further, since $X_f>a$ and $Y_f>b$, $f$ is rejected on both sides.
    Finally, the reverse map gives $X_f-Y_f=\frac{(a-b)(p-Y_f)}{p-b}\ge0$, since $a\ge b$ and $Y_f\le p$.

    \bigskip
    \underline{$\ForwardKernel$-2(c)(ii) $\leftrightarrow$ $\ReverseKernel$-2(b)(ii)}.~
By induction, the two recursive forward-kernel calls succeed with probabilities $\frac{p}{\sigma_{e,F(u)}}$ and $\frac{p}{\sigma_{f,F(u)+e}}$ on a uniform input, and each produces a uniform output conditional on success.
    Multiplying by the branch probability $1-\frac{\sigma_{e,F(u)}}{\sigma_0}$ gives $\rho(u)$.
    The draw $Y_f\sim\Unif(b,p)$ has density $\frac{1}{p-b}$ over $(b,p)$.
    Thus, this branch contributes density relative to the uniform law
\begin{align*}
        \frac{1}{p-b} \int_p^{\sigma_0}\rho(u)\,du
         & =\frac{1}{p-b} \int_p^{\sigma_1}\left(1-\frac{\sigma_1}{\sigma_0}\right)\frac{p}{\sigma_1}\frac{p}{u}\,du
        +\frac{1}{p-b} \int_{\sigma_1}^{\sigma_0}\left(1-\frac{u}{\sigma_0}\right)\frac{p}{u}\frac{p}{\sigma_1}\,du
        \tag{by \eqref{eq:poly-exchange-identities}}                                                                 \\
         & =\frac{p}{\sigma_0}(1-\lambda),
\end{align*}
which equals the required adjoint factor $\frac{p}{\sigma_0}$ times the probability $1-\lambda$ of choosing $\ReverseKernel$-2(b)(ii).
Within this branch, the contribution from each $u$ is proportional to $\rho(u)$, matching the draw in $\ReverseKernel$-2(b)(ii).
    Therefore, the adjoint identity holds in this case.

    Further, the two recursive calls change the filtration as $F(u)\rightarrow F(u)+e\rightarrow F(p)+e$ and $\ReverseKernel$-2(b)(ii) reverses these changes.
    By induction, both calls preserve the survivor set on $S-f$.
    Since this set is unchanged, both calls move each remaining coordinate in the same direction, proving monotonicity.
    Finally, we have $X_f>p$ and $b<Y_f<p$, so $f$ is rejected on both sides and $X_f>Y_f$.

    \bigskip
    \underline{$\ForwardKernel$-2(c)(i) $\leftrightarrow$ $\ReverseKernel$-2(c)}.~
By induction, the recursive forward-kernel call succeeds with probability $\frac{p}{\sigma_{e,F(u)}}$ on a uniform input and produces a uniform output conditional on success.
    The branch is chosen with probability $\frac{\sigma_{e,F(u)}}{\sigma_0}$.
    Their product is $\frac{p}{\sigma_0}$ as required.
    $\ReverseKernel$-2(c) applies the inverse scalar map and the recursive reverse kernel.
    Further, we have $X_f=Y_f>p$, so $f$ is rejected on both sides and the coordinate inequality holds with equality.

    \paragraph{Complexity.}
We first bound the expected number of recursive calls on a uniform input.
    For the forward kernel, we show that a uniform input gives at most one direct recursive call in expectation.
    This is trivially true for base cases.
    For non-base cases, only $\ForwardKernel$-2(a)(i) makes no call, and only $\ForwardKernel$-2(c)(ii) can make two calls.
    It suffices to compare their probabilities.

    Branch $\ForwardKernel$-2(a)(i) corresponds to $\frac{\sigma_1}{\sigma_0}b<X_f\le a$, an event of probability $\frac{1}{t}(a-\frac{\sigma_1}{\sigma_0}b)$.
    The second call in $\ForwardKernel$-2(c)(ii) occurs if this branch is chosen and the first call succeeds.
    Given any $X_f=u>p$, this has probability
    \[
        \left(1-\frac{\sigma_{e,F(u)}}{\sigma_0}\right)\frac{p}{\sigma_{e,F(u)}}
        =\frac{p}{\sigma_{e,F(u)}}-\frac{p}{\sigma_0}.
    \]

    By \eqref{eq:poly-exchange-identities}, $\sigma_{e,F(u)}=\sigma_1$ for $p<u\le\sigma_1$, $\sigma_{e,F(u)}=u$ for $\sigma_1<u<\sigma_0$, and $\sigma_{e,F(u)}=\sigma_0$ for $u\ge\sigma_0$.
    Thus, the probability of making two calls is
    \[
        \frac1t\left(
        \int_p^{\sigma_1}\left(\frac{p}{\sigma_1}-\frac{p}{\sigma_0}\right)du
        +\int_{\sigma_1}^{\sigma_0}\left(\frac{p}{u}-\frac{p}{\sigma_0}\right)du
        \right)
        =\frac1t\left(a-b-p^2\left(\frac1{\sigma_1}-\frac1{\sigma_0}\right)\right)
    \]

    Since $\sigma_1 \le \sigma_0$ and $b\ge0$, this is at most $\frac{1}{t}(a-\frac{\sigma_1}{\sigma_0}b)$, the probability of making no call.

    Conditional on $X_f$ and the branch choice, each first recursive call has a uniform input.
    By the adjoint identity, each second call also has a uniform input conditional on the first succeeding.
    Thus, the expected number of calls at each recursion depth is at most one.
    Since every recursive call reduces the size of $S$, the expected total number of forward calls is at most $|S|+1$.

    Next, consider the reverse kernel.
    The same branch correspondence used in the adjoint proof matches entire successful forward executions with reverse executions, with the same number of recursive calls.
    Applying the inductive adjoint identity at each call gives the same factor $\frac{p}{\sigma_{e,F}}$ for these matched executions.
    Hence, on a uniform input, the expected number of reverse calls equals the expected number of forward calls conditional on success.
    Since the success probability is $\frac{p}{\sigma_{e,F}} \ge \frac{1}{\e}$ by $p \ge \frac{1}{\e}$ and $\sigma_{e,F} \le 1$, the forward-call bound gives an expected number of reverse calls at most $\e(|S|+1)$.

    It remains to bound the work within each call.
    Most steps are elementary.
    Finding the biggest element $f$ uses polynomially many comparison queries.
    The draw from the density proportional to $\rho$ in $\ReverseKernel$-2(b)(ii) takes constant expected time by standard inverse-transform and rejection sampling.
    Finally, $\sigma_{e,F}$ can be computed in polynomial time using independence queries by checking spanning at $p$ and at each release time in $F$.
\end{proof}

\section*{Declaration of AI Usage}
The algorithm and analysis were developed through extensive interactions with ChatGPT Astra.
The author directed the model to explore different constructions and made several judgment calls, including focusing the exploration on simulating an offline distribution of accepted elements, observing that the proofs of \Cref{lem:distribution-change} and \Cref{thm:simulation} are essentially algorithmic, and pursuing a full resolution after the model produced partial results.\footnote{These include the same $3.1462$-competitive algorithm as \cite{ChanLW:arXiv:2026}, a $2.9654$-competitive algorithm for a family of non-linear matroids, and on 19 September 2026, the computationally inefficient version of the $\e$-competitive algorithm.}
The author has verified the correctness of the algorithm and proofs and assumes full responsibility for the contents.

\bibliographystyle{plainnat}
\bibliography{matroid-secretary}

\end{document}